\documentclass[10pt,twocolumn,letterpaper]{article}

\usepackage{ijcb}
\usepackage{times}
\usepackage{epsfig}
\usepackage{graphicx}
\usepackage{amsmath}
\usepackage{amssymb}
\usepackage{amsthm}
\usepackage{stmaryrd}
\usepackage{algorithm}
\usepackage{algorithmic}
\usepackage{multirow}
\usepackage{paralist}
\usepackage{float}

\newtheorem{definition}{Definition}
\newtheorem{theorem}{Theorem}

\newcommand{\encbr}[1]{{\ensuremath\llbracket #1 \rrbracket}}
\newcommand{\enc}{\encbr}

\newcommand{\N}{\ensuremath \mathbb{N}}
\newcommand{\Z}{\ensuremath \mathbb{Z}}
\newcommand{\A}{\ensuremath \mathcal{A}}

\newcommand{\C}{\ensuremath \mathcal{C}}
\newcommand{\D}{\ensuremath \mathcal{D}}

\newcommand{\X}{\ensuremath \mathcal{X}}

\newcommand{\PRF}{\ensuremath \operatorname{PRF}}

\usepackage{color}

\newcommand{\jaro}[1]{{\color{blue} \textbf{Jaro says: #1}}}

\usepackage{setspace,calc}
\newlength{\wleft}
\newlength{\wmid}
\newlength{\wright}
\makeatletter
\renewcommand{\paragraph}{%
  \@startsection{paragraph}{4}%
  {\z@}{0.5ex \@plus 1ex \@minus .2ex}{-1em}%
  {\normalfont\normalsize\bfseries}%
}
\makeatother

\usepackage[breaklinks=true,letterpaper=true,bookmarks=false]{hyperref}

\ijcbfinalcopy %

\ifijcbfinal\fi
\begin{document}

\title{Privacy-Preserving Population-Enhanced Biometric Key Generation from Free-Text Keystroke Dynamics -- Extended Version\thanks{This work was supported in part by DARPA Active Authentication grants FA8750-12-2-0201 and FA8750-13-2-0274 and NYIT ISRC 2012 and 2013 grants. The views, findings, recommendations, and conclusions contained herein are those of the authors and should not be interpreted as necessarily representing the official policies or endorsements, either expressed or implied, of the sponsoring agencies or the U.S. Government.}}

\author{Jaroslav \v Sed\v enka\thanks{J. \v{S}ed\v{e}nka (sedenka@mail.muni.cz) is a student of Faculty of Science, Masaryk University, Czech Republic. This work was done while visiting the New York Institute of Technology.} \\
{\footnotesize New York Inst. of Technology}\\
{\tt \footnotesize jsedenka@nyit.edu}
\and
Kiran S.~Balagani\\
{\footnotesize New York Inst. of Technology}\\
{\tt \footnotesize kbalagan@nyit.edu}
\and
Vir Phoha\\
\footnotesize Louisiana Tech University\\
{\tt \footnotesize phoha@latech.edu}
\and
Paolo Gasti\\
{\footnotesize New York Inst. of Technology}\\
{\tt \footnotesize pgasti@nyit.edu}
}

\maketitle
\thispagestyle{empty}

\begin{abstract}
Biometric key generation techniques are used to reliably generate 
cryptographic material from biometric signals. Existing constructions require 
users to perform a particular activity (e.g., type or say a password, or provide 
a handwritten signature), and are therefore not suitable for generating keys 
continuously. 
In this paper we present a new technique for biometric key generation from 
free-text keystroke dynamics. This is the first technique suitable for 
continuous key generation. Our approach is based on a scaled 
parity code for key generation (and subsequent key reconstruction), and can be 
augmented with the use of population data to improve security and reduce key 
reconstruction error. In particular, we rely on linear discriminant analysis 
(LDA) to obtain a better representation of discriminable biometric signals. 

To update the LDA matrix without disclosing  user's biometric information, we 
design a provably secure privacy-preserving protocol (PP-LDA) based on 
homomorphic encryption. Our biometric key generation with PP-LDA was evaluated on a dataset 
of 486 users. We report equal error rate around 5\% when using LDA, and below 7\% without 
LDA.
\end{abstract}

\section{Introduction}

Biometric Key Generation (BKG) harnesses biometric signals to protect 
cryptographic keys against unauthorized access.
Freshly-generated keys are {\em committed} using biometric information; 
subsequently, the same biometric signal (or a ``close enough'' signal) is used 
to reconstruct (i.e., {\em decommit}) a key. 

BKG offers unique and appealing features. Unlike passwords, biometric 
information (and the corresponding key) is tied to a particular user, and as 
such cannot be easily disclosed or stolen (e.g., via shoulder-surfing). 
Easy-to-remember passwords provide only marginal security, while strong 
passwords are difficult to remember. Ideally, by relying on high-entropy and 
consistent biometric signals, BKG is a good candidate for replacing password-based techniques, because it offers a good tradeoff between usability and security. 

Although originally conceived for physical biometrics~\cite{fuzzy_commitment}, 
such as fingerprints and iris, BKG techniques have recently been applied to {\em 
behavioral biometrics}. These techniques include key generation using 
voice~\cite{mon01}, handwritten signatures~\cite{fen02,fre06} and 
keystroke dynamics~\cite{keystroke-BKG}. Due to the inherent 
variability of behavioral signals, current approaches require users to perform 
a specific activity while these signals are collected. For example,  
techniques based on voice recognition require users to pronounce the same 
sentence (i.e., a passphrase) for both committing and decommitting a key; BKG 
based on keystroke dynamics {\em augments} password-based systems by introducing 
additional entropy while the user is typing her password.

This work is the first to introduce BKG on {\em free-text} input. (In free-text 
setting, users are allowed to type or say any text.)  Free-text BKG allows 
periodic key generation using behavioral data collected continuously, since 
collection of biometric signals does not interfere with regular user activity. 
Therefore, free-text based systems can capture biometric signals over a long 
period of time, providing better accuracy and security.

\paragraph{Contributions.}
We propose a novel BKG technique that builds on the fuzzy commitment schemes of 
Juels \etal~\cite{fuzzy_commitment}. 
Our work represents a further step towards a formalization of 
security of BKG techniques. We define new and more realistic requirements for biometric 
signals. In particular, while fuzzy commitments of Juels \etal are secure 
under the unrealistic assumption that all biometric features are uniformly 
distributed, our commitments are provably secure without any assumption on the 
distribution of the biometric features. Instead, we assume that 
biometric signals are an instantiation of an {\em unpredictable function} \cite{unpredictable}, which 
is a well-understood cryptographic tool. 

Furthermore, we extend the error-correcting code to work on arbitrary 
biometrics.
We then improve commitment/decommitment performance by relying on population 
data. We evaluate the feasibility of using Linear Discriminant 
Analysis~\cite{fis36} (LDA) for improving the EER of our technique. Since LDA 
requires data from multiple users, we design a privacy-preserving protocol 
(PP-LDA) that allows users to compute LDA parameters without disclosing their 
biometric signals. The protocol involves two untrusted parties: an {\em enrollment 
server} (ES) and a {\em matrix publisher} (MP). To our knowledge, this work is 
the first to use LDA (and, consequently, PP-LDA) for the purpose of biometric 
key generation.

By using keystroke and digraph latency information, we were able to achieve 
5.5\% false accept rate (FAR) and 3.6\% false reject rate (FRR).

\paragraph{Organization.}
The rest of the paper is organized as follows. In Section~\ref{sec:background} 
we overview the tools and the security model used in this paper. Section 
\ref{sec:techniques} introduces our modification of Fuzzy Commitment, and 
presents our PP-LDA protocol. In Section~\ref{sec:security} we analyze the 
security of our BKG technique and of PP-LDA. Experimental evaluation is 
presented in Section \ref{sec:evaluation}. We summarize related work in 
Section~\ref{sec:related} and conclude in Section~\ref{sec:conclusion}.

\section{Background}
\label{sec:background}

	\paragraph{LDA.}
	\label{sec:background_lda}
	
Linear Discriminant Analysis \cite{fis36} (LDA) is a well known supervised feature 
extraction method. The goal of LDA is to derive a new (and possibly compact) set 
of features from the original feature set, such that the new set provides 
increased class-discriminability.
Formally, LDA finds direction vectors for projections that maximize linear 
separability between classes (`users' in our case). Let $v$ denote the number of 
users, $X_i$ be the $m_{i} \times n$ matrix of data samples containing $m_{i}$ $n-
$dimensional training samples of user $i$. The mean of $X_{i}$ is denoted by the 
row vector $\mu_{i}$ and $\mu$ is the mean of all $\mu_{i}$-s. Let $C_{m}$ denote 
the $m \times m$ centering matrix.

LDA finds the transformation matrix $W = [w_{1};\ldots;w_{n-1}]$ that maximizes the objective function $J(W) = \frac{|W^TS_{b}W|}{|W^TS_{w}W|}$, where the {\em scatter between} term `$S_{b}$' is calculated as $\sum_{i=1}^{v} (\mu_i - \mu)^T(\mu_i - \mu)$ and the {\em scatter within} term `$S_{w}$' is calculated as $\sum_{i=1}^{v}X_{i}^{T}C_{m_i}X_{i}$. It can be easily shown that the transformation matrix $W$ is the solution of the generalized eigenvalue problem $S_{B}W = \Lambda S_{W}W$. After the transformation matrix is obtained, the new $(n-1)$ features are calculated as $XW$.

\paragraph{\bf Fuzzy Commitments and BKG.}
Fuzzy commitments~\cite{fuzzy_commitment} use error correcting codes to construct commitments from noisy information, e.g., biometric signals.
Features are extracted from raw signals (e.g., minutiae from fingerprint images); then, each feature is represented using a single bit, therefore 
constructing an $n$-bit vector (where $n$ is the number of features) for each 
sample. Let $C \subseteq \{ 0, 1 \}^n$ be a group error-correcting code. To 
commit a codeword $c\in{C}$ using biometric $x = (x_1, \ldots, x_n)$, 
the user computes commitment $(\mathrm{H}(c), \delta = x \oplus c)$; the biometric 
key is computed as $k = \mathrm{H}'(c)$, where $\mathrm{H}$ and $\mathrm{H}'$ are 
two collision resistant hash functions.

To decommit the biometric key using biometric sample $y = (y_1, \ldots, y_n)$, 
the user computes codeword $c' = {\sf decode}(y \oplus \delta)$. If $\mathrm{H}(c') = 
\mathrm H(c)$, then the biometric key is computed as $k = \mathrm{H}'(c')$. Otherwise, no information about the key is revealed. Biometric vector $y$ decommits $c$ iff the 
error vector $e = x - y$ decodes to the zero codeword. 

Juels \etal~\cite{fuzzy_commitment} prove that the adversary cannot reconstruct $c$ from a commitment under the assumption that $\mathrm{H}$ and $\mathrm{H}'$ are collision-resistant functions and that $x$ is uniformly distributed in $\{ 0, 1 \}^n$.

\paragraph{\bf Unpredictable functions.}

In contrast with the technique in~\cite{fuzzy_commitment}, we model biometric 
features as unpredictable functions~\cite{unpredictable}. This captures the idea 
that a user's biometric is {\em difficult to guess}. Informally, an 
unpredictable function $f(\cdot)$ is a function for which no efficient adversary 
can predict $f(x^*)$ given $f(x_i)$ for various $x_i \neq x^*$. Formally:

\begin{definition}
A function family $(\C, D, R, F)$ for $\{f_c(\cdot): D \rightarrow R\}_{c \leftarrow \C}$ is unpredictable if for any efficient algorithm $\A$ and auxiliary information $z$ we have:
\[
	Pr[(x^*, f_c(x^*) \leftarrow \A^{f_c(\cdot)}(z) \text{\; and\; } x^* \not\in Q] \leq \operatorname{negl}(\kappa)
\]
where $Q$ is the set of queries from $\A$, $\kappa$ is the security parameter and $\operatorname{negl}(\cdot)$ is a negligible function.
\label{def:unpredictable}
\end{definition}

\paragraph{\bf Homomorphic Encryption.} Our PP-LDA construction uses a 
semantically secure (public key) additively homomorphic encryption scheme. Let $
\encbr{m}$ indicate the encryption of message $m$ using a homomorphic encryption 
scheme. (We omit public keys in our notation, since there is a single 
public/private keypair generated by MP.) We have that $\enc{m_1}\cdot\enc{m_2} = 
\enc{m_1 + m_2}$, which also implies that $\enc{m}^a = \enc{m \cdot a}$. While 
any encryption scheme with the above properties 
suffices for the purposes of this work, the 
construction due to Damg{\aa}rd \etal~\cite{dam08,dam08cor} (DGK hereafter) is 
of particular interest here because it is fast and produces small 
ciphertexts. 	

Fully-homomorphic encryption (FHE) can also be used to instantiate our PP-LDA 
protocol. However, due to the severe performance penalty associated with FHE, we design a protocols that requires only additively homomorphic encryption.

\paragraph{Security Model and Definitions.}
Our protocols are secure in the presence of semi-honest (also known as 
honest-but-curious or passive) participants. In this model, while participants 
follow prescribed protocol behavior, they might try to learn additional information 
beyond that obtained during normal protocol execution.

We use the term {\em adversary} to refer to insiders, i.e., protocol participants. 
This includes the case when one of the parties is compromised. Outside adversaries, e.g., those 
who can eavesdrop on the communication channel, are not considered since their actions can be 
mitigated via standard network security techniques (e.g., by performing all communication over SSL).	

\section{Our Techniques}
\label{sec:techniques}

In this section we introduce our BKG construction. To generate a cryptographic key, the user selects a random codeword $c$ of 
length $n$ from a custom error-correcting code $C$, and uses it to derive 
a cryptographic key as $k = \PRF_c(z)$, where $\PRF$ is a pseudorandom function 
and $z \neq 0$ is either a system-wide public constant or a user-provided 
pin/password for added security. (In our security analysis we assume that the adversary knows $z$) $c$ is then protected using the user's biometrics as follows. 
After collecting keystroke data, we extract $n$ keystroke and digraph features 
$x = (x_1, \ldots, x_n)$, which are then discretized and scaled by their 
standard deviation. The user then computes $\delta = (x-c) = (x_1-c_1, \ldots, x_n-c_n)$ 
and publishes commitment $(\PRF_c(0), \delta)$.

The user can reconstruct the cryptographic key given public parameters $(\PRF_c(0), \delta)$, a biometric signals and possibly a pin/password $z$ as follows. 
The user extract keystroke/digraph features from the sample; then she
constructs vector $y=(y_1,\ldots,y_n)$. Finally, she computes $c' = {\sf decode}(y-\delta)$. 
If $\PRF_{c'}(0) = \PRF_c(0)$, then $k=\PRF_{c'}(z)$ is the correct key with 
overwhelming probability.
In the following, we provide further details on our construction.

\paragraph{\bf Our Construction.} In our scheme, each feature is discretized and mapped to the range $[0,2^d-1]$. In other words, codeword symbols are elements of 
$\Z_{2^d}$. Discretization is performed as:
\[
{\sf discretize}_{d,{\sf F}}(x_j) = \left\lfloor (2^d-1)\left(\frac{ x_j-{\sf min_F}}{{\sf max_F}-{\sf min_F}}\right)\right\rfloor
\]
where $\sf F$ is the feature being discretized, $x_{j}$ is an instance in $\sf F$, 
${\sf min_F}$ is the minimum value of $\sf F$, and ${\sf max_F}$ is the maximum 
value of $\sf F$. The $d$ parameter controls the number of cells a feature is 
discretized into. Therefore, higher the $d$, the lower the potential loss of 
information due to discretization. When  $x_j > {\sf max_F}$ or $x_j < {\sf min_F}$, we set ${\sf discretize}_{d,F}(x_j)$ to $2^d-1$ and $0$, respectively.

Distance between two codewords is not defined via the usual Hamming distance. In 
fact, Hamming distance captures well the ``similarity'' between a bit string and 
its perturbed version when all bits in the string have the same probability of 
being affected by an error. In our setting this is not the case, 
because the least significant bits of each feature instance have higher 
probability of being altered. 

\begin{figure}[t]\resizebox{0.46\textwidth}{!}{
\begin{tabular}{lllll}
\cline{2-5}
\multicolumn{1}{r}{$c =$} & \multicolumn{1}{|l}{\texttt{\textbf{1101}000}} & \multicolumn{1}{|l}{\texttt{\textbf{01010}00}} & \multicolumn{1}{|l}{\texttt{\textbf{01}00000}} & \multicolumn{1}{|l|}{\texttt{\textbf{100}$p$000}} \\ \cline{2-5} 
                       & \multicolumn{1}{c}{$c_1$}     & \multicolumn{1}{c}{$c_2$}     & \multicolumn{1}{c}{$c_3$}     & \multicolumn{1}{c}{$c_4$}     
\end{tabular}}
\caption{Example of a SPC codeword. Each codeword symbol $c_i$ for $i<n$ ends with $l_i$ zero bits. (We write the least significant $l_i$ bits of each symbol using non-bold typeface.) The last symbol ends with $l_n - 1$ zero bits, preceded by one parity bit, denoted with $p$ in the figure. }
\label{fig:codeword}
\end{figure}

\if{false}
	In practice, small variations 
	in measurements can lead to a large Hamming distance between the correct and the 
	corrupted word. Similarly, large variations between samples might result in a 
	small Hamming distance. As an example, the Hamming distance between the 
	``natural'' binary representation of 191 and 192 is 7, while the Hamming distance 
	between 191 and 255 is 1. 
\fi

Therefore, we instantiate fuzzy commitments using a custom-designed error-correcting code inspired by codes in the Lee 
metric~\cite{lee58}. Let us define Lee weight and Lee distance as follows:
\begin{definition}[Lee weight]
The \emph{Lee weight} of element $x \in \Z_{2^d}$ is defined as $w_L(x) = abs(x')$, such that $x' \equiv x \bmod{2^d}$ and $-2^{d-1} < x' \leq 2^{d-1}$. 
The \emph{Lee weight} of vector $x = (x_1,\ldots, x_n) \in (\Z_{2^d})^n$ is defined as the sum of Lee weights of its elements, i.e., $w_L(x) = \sum_{i=1}^n{w_L(x_i)}$.
\end{definition}

\begin{definition}[Lee distance]
The \emph{Lee distance} of vectors $x, y \in \Z_{2^d}$ is the Lee weight of their difference, i.e., $d_L(x, y) = w_L(x - y)$. 
\end{definition}
\if{false}
Note that in $\Z_2$, the Lee weights is the same as the Hamming weight.
\fi

We consider each feature vector a (possibly perturbed) codeword of a code in Lee metric, and we 
use the {\em Lee weight} as a metric for distance between feature vectors. We 
refer to individual elements in a codeword (i.e., individual features) as {\em 
symbols}. 
Features are scaled by the standard deviation and discretized as  $s'_i = {\sf discretize}_{d, {\sf F}}(\sigma_i \cdot \kappa)$, where $\sigma_i$ is the standard deviation of feature $\sf F$. Then $s'_i$ is mapped to the closest power of two, which we indicate as $s_i$. The $l_i$ least significant bits of $c$ are zero in all codewords. (The last symbol has $l_n-1$ bits set to zero since the least significant bit is used for parity.)

Existing codes in the Lee metric only guarantee correct decoding when the Lee 
weight of the error is a small multiple of the number of codeword 
symbols~\cite{wuh08,rot06}. However, in our setting the number of codeword symbols 
is relatively small (i.e., between 20 and 100), while the domain for each feature 
is large (integers between $0$ and $2^{24}-1$ in our experiments). 
Therefore, we use group error-correcting code, which we refer to as a \emph{scaled 
parity code} (SPC). 
Let $n$ denote the number of features, $d$ be the discretization parameter, $\kappa$ the security parameter. A vector $c = (c_1, \ldots, c_n) \in (\mathbb{Z}_{2^{d}})^n$ is a SPC codeword iff it satisfies the following two 
conditions: (a) for all $i \in \{1, \ldots, n\}$, $s_i$ divides $c_i$; and (b) $
\sum_{i=1}^n c_i/s_i \equiv 0 \mod{2}$ (i.e., {\em parity condition}). Figure \ref{fig:codeword} shows the structure of a sample codeword.

If the parity condition is not met during decoding, we select $c'_k$ such that $|y_k-x_k|$ is assumed to be largest (after normalization) among all $|y_i-x_i|$. The parity is corrected by adding (subtracting) $s_k$ to $c'_k$ if $(y_k - \delta) - c'_k$ is positive (negative, respectively).

Our SPC is designed to guarantee that only feature vectors ``close'' to the user's 
template decommit to the correct codeword. 
When the number of codeword symbols is either one or two, the SPC algorithm 
decodes vectors to the closest codeword in the Lee metric.

\begin{theorem} \label{thm:lee}
Let $n \in \N$ be the number of features, $d \in \N$ the discretization parameter, $C \subset \Z_{2^d}^n$ a scaled parity code with scaling $s_1, ..., s_n$, $c \in C$ be a codeword and $\gamma = c + \epsilon$ where $\epsilon \in \Z_{2^d}^n$ is the error. If the sum of the two biggest relative errors is smaller than one, Algorithm \ref{fig:decoding} decodes $\gamma$ to $c$.
\end{theorem}

\if{false}
\begin{proof}
Left to the reader as an exercise.
\end{proof}
\fi

\if{false}
In particular, two vectors are close if at most two of their features have a distance larger than the threshold.
More precisely, let $x = (x_1, \ldots, x_n)$ be the committing vector, and the decommitting vector $y = (y_1, \ldots, y_n)$. $y$, perturbed with error $e = (e_1, \ldots, e_n)$ as $x = y + e$, decommits correctly if and only if $max_i(e_i/s_i) + max_{j \in \{1, \ldots, n\} \backslash \{i\}}(e_j/s_j) < 1$. The complete process is shown in Algorithm~\ref{fig:decoding}.
\fi

\begin{algorithm}[h]
\caption{\sc{Decommit Codewords in SPC}}	
\label{fig:decoding} \small
	\begin{algorithmic}[1]
	 \item on input $\gamma = (\gamma_1, \ldots, \gamma_n)$, \\scaling factors $s = (s_1, ..., s_n)$:
	 \FOR{each feature $i$}
	  \STATE $e_i = \gamma_i \bmod s_i$  	\hspace{0.1cm}{\tt //} $e_i$\; {\tt is error on feature } $i$
	  \IF{$e_i/s_i > 1/2$}
	   \STATE $e_i = e_i-s_i$
	  \ENDIF
	 \ENDFOR
	\STATE $c' = \gamma - e$						\hspace{2.15cm}{\tt //subtract error}
	\STATE $p = \sum_{i=1}^n(c'_i/s_i)$ 	\hspace{1.25cm}{\tt //check parity}
	\IF{$p \equiv 1 \pmod{2}$}
	 \STATE $k = argmax_i(|e_i/s_i|)$\\			{\tt //feature index with max. relative error}
	 \STATE $c'_k = c'_k + sign(e_k)s_k$
	\ENDIF
	\RETURN $c' = (c_1, \ldots, c_n)$
\end{algorithmic}
\end{algorithm}

\label{sec:protocol_lda}

\if{false}
The discriminability power and associated entropy of biometric features is 
crucial for minimizing {\em false accept rate} (FAR) and {\em false reject rate} (FRR) of the fuzzy 
commitments. In order to improve feature quality, we enhance the technique 
presented in Section~\ref{sec:ecc} using LDA (see Section~
\ref{sec:background_lda}). 
To the best of our 
knowledge, this is the first work that uses LDA to improve BKG performance. 
LDA can be used to improve performance of BKG based on other biometric 
modalities, which makes this contribution of independent interest.

While the BKG scheme introduced in the previous 
section does not require any additional information besides a user's biometric 
samples, LDA transformation must be computed over population biometric data. A 
straightforward way to implement the transformation is, therefore, to disclose 
biometric measurements for all users. This is clearly unacceptable, because: (1) if the same biometric data is also used for 
authentication, by disclosing individuals' samples an adversary can easily 
impersonate any of the users; and (2) more related to BKG, once users' samples 
have been disclosed, an adversary can trivially use these samples to decommit 
biometric keys. 
\fi

\paragraph{\bf Privacy-Preserving LDA.}
\label{sec:pp-lda}
To avoid releasing  individual users' biometrics, we designed a three-party 
\emph{privacy-preserving linear discriminant analysis} protocol,  illustrated in Figure \ref{fig:pp-lda}. The protocol is 
executed when new users enroll. The other two parties involved are the 
{\em enrollment server} (ES) and the {\em matrix publisher} (MP). The user generates biometric samples, encrypts them under the MP's public key and sends them to ES. ES stores the user's samples in encrypted form and computes, in conjunction with the user, the updated {\em encrypted} scatter within and scatter between matrices. The matrices are then sent to MP, which decrypts them and publishes the corresponding LDA matrix. We assume that MP does not collude with either ES or any user, as MP can 
decrypt any encrypted message. Interaction between the user and ES/MP is 
necessary only during enrollment. After that, the user is able to generate 
biometric keys using local data. 

When using LDA for fuzzy commitments, the transformation matrix is also necessary 
in order to recover the key (i.e., to decommit). LDA matrix and scaling factors are not user-specific, 
therefore they only reveal information about the overall population. 
However, to take full advantage of the population data -- especially to 
minimize FAR -- users should update their LDA matrix periodically 
to include data from new users. (Update interval depends on a number of factors, such as the number of users, how many users join the system in a given time interval, etc.) After the matrix is updated and published, biometric keys 
must be re-generated since a biometric signal used with a different transformation matrix cannot not be used to reconstruct the key.
During enrollment, each user $u$ holds a $m_u \times n$ matrix $X_u$ containing her $m_u$ training samples. The LDA algorithm creates a linear transformation that transforms the feature vectors to a space where the Fisher's Discriminant is maximal. Samples of all enrolled users are required to compute the transformation. 

Our protocol guarantees that the ES does not learn any information about the users' input. Similarly, a new user do not learn information about the biometrics of existing users.

When a new user joins the system, an updated version of the LDA matrix is 
generated. In order to prevent the adversary from extracting information on the new 
user by comparing two consecutive LDA matrices, the ES provides its output to the MP only after a pre-determined number of users $w\gg1$ have joined the system. This way, the adversary can only learn aggregate information of $w$ users.

\begin{figure*}[t]
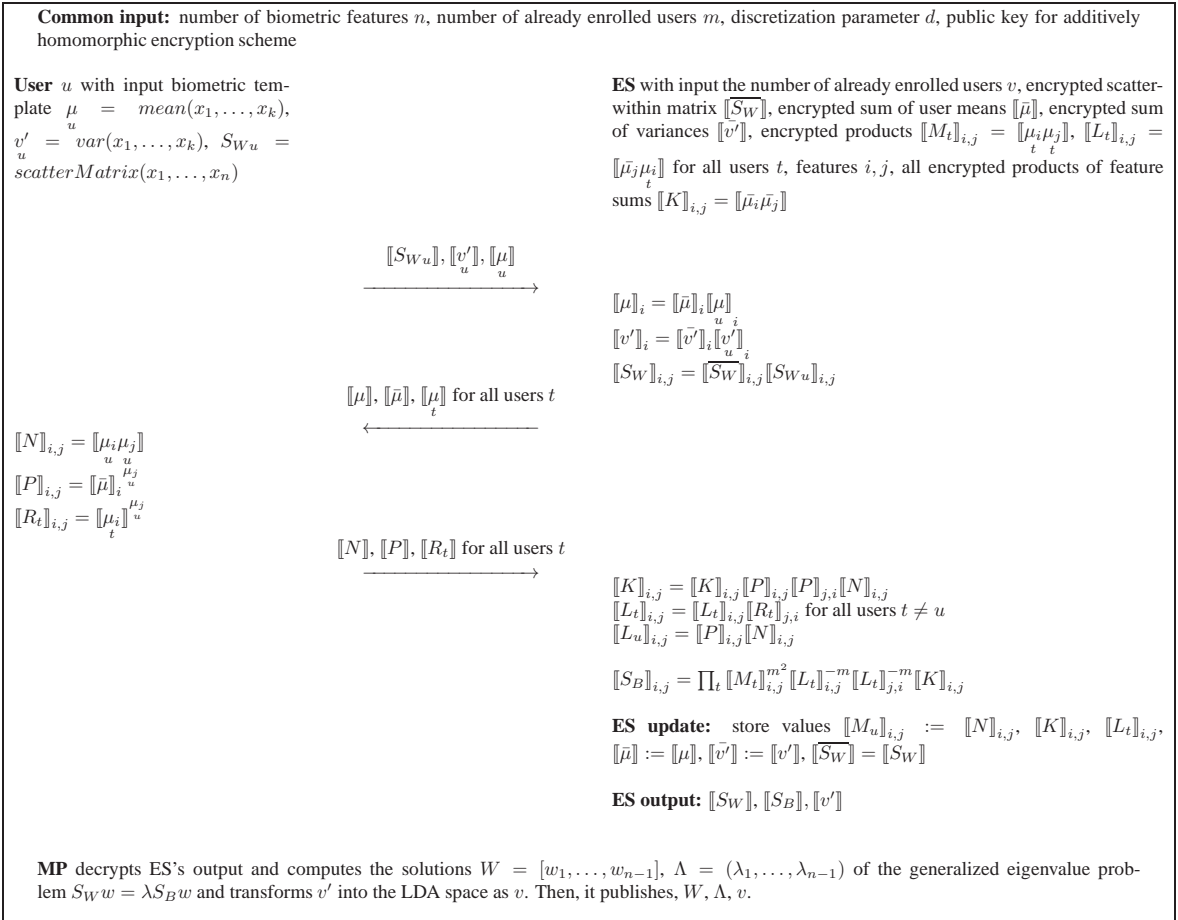

\centering
\resizebox{0.9\textwidth}{!}{
	\begin{tabular}{|p{\wleft} p{\wmid} p{\wright}|}
	\hline
		\multicolumn{3}{|c|}{ 
			\parbox{\wleft+\wmid+\wright}{ \vspace{0.1cm}{\bf Common input:} number of biometric features $n$, number of already enrolled users $m$, discretization parameter $d$, public key for additively homomorphic encryption scheme} }\\

		 \vspace{0.2cm}  {\bf User} $u$ with input biometric template $\underset{u}{\mu}=mean(x_1, \ldots, x_k)$, $\underset{u}{v'}=var(x_1, \ldots, x_k)$, $S_{Wu} = scatterMatrix(x_1, \ldots, x_n)$ & & 
		\vspace{0.2cm}
		{\bf ES} with input the number of already enrolled users $v$, encrypted scatter-within matrix $\enc{\overline{S_W}}$, encrypted sum of user means $\enc{\bar\mu}$, encrypted sum of variances $\enc{\bar{v'}}$, encrypted products $\enc{M_t}_{i,j} = \enc{\underset{t}{\mu_i} \underset{t}{\mu_j} }$, $\enc{L_t}_{i,j} = \enc{\bar{\mu_j} \underset{t}{\mu_i}}$ for all users $t$, features $i,j$, all encrypted products of feature sums $\enc{K}_{i,j} = \enc{\bar{\mu_i} \bar{\mu_j}}$ \\

		\vspace{0.05cm}
		\parbox{\wleft}{

		} &

		\parbox{\wmid}{ \centering \ \\ \ \\ $\enc{S_{Wu}}, \enc{\underset{u}{v'}}, \enc{\underset{u}{\mu}}$ \\$\xrightarrow{\hspace*{3cm}}$ }
		& \\

		& & 				
		\parbox{\wright}{
			$\enc{\mu}_i = \enc{\bar\mu}_i\enc{\underset{u}{\mu}}_i$ \\
			$\enc{v'}_i = \enc{\bar{v'}}_i\enc{\underset{u}{v'}}_i$ \\
			$\enc{S_W}_{i,j} = \enc{\overline{S_W}}_{i,j}\enc{S_{Wu}}_{i,j}$
		} \\	

		&	
		\parbox{\wmid}{ 
			\centering $\enc{\mu}$, $\enc{\bar\mu}$, $\enc{\underset{t}{\mu}}$ for all users $t$
			\\ $\xleftarrow{\hspace*{3cm}}$ 
		}
		& \\

		\parbox{\wleft}{
			$\enc{N}_{i,j}   =  \enc{\underset{u}{\mu_i} \underset{u}{\mu_j} }$  \\
			$\enc{P}_{i,j}   = {\enc{\bar{\mu}}_i}^{\underset{u}{\mu_j} }$ \\
			$\enc{R_t}_{i,j} = {\enc{\underset{t}{\mu_i}}}^{\underset{u}{\mu_j} }$  
		} & &
		\\ 

		&	
		\parbox{\wmid}{ \centering $\enc{N}$, $\enc{P}$, $\enc{R_t}$ for all users $t$ \\ $\xrightarrow{\hspace*{3cm}}$ }
		& \\

		& & \parbox{\wright}{

			$\enc{K}_{i,j} = \enc{K}_{i,j} \enc{P}_{i,j} \enc{P}_{j,i} \enc{N}_{i,j}$ \\				
			$\enc{L_t}_{i,j} = \enc{L_t}_{i,j} \enc{R_t}_{j,i}$ for all users $t \not= u$ \\
			$\enc{L_u}_{i,j} = \enc{P}_{i,j} \enc{N}_{i,j}$ \\		
			
			$\enc{S_B}_{i,j} = \prod_t \enc{M_t}_{i,j}^{m^2} \enc{L_t}_{i,j}^{-m} \enc{L_t}_{j,i}^{-m} \enc{K}_{i,j} $ \\

			{\bf ES update:} store values $\enc{M_u}_{i,j} := \enc{N}_{i,j}$, $\enc{K}_{i,j}$, $\enc{L_t}_{i,j}$, $\enc{\bar\mu} := \enc{\mu}$, $\enc{\bar{v'}} := \enc{v'}$, $\enc{\overline{S_W}} = \enc{S_W}$ \\
			
			{\bf ES output:} $\enc{S_W}$, $\enc{S_B}, \enc{v'}$ \\
			
		}
		\\

	\multicolumn{3}{|c|}{  
		\parbox{\wleft+\wmid+\wright}{ \ \\  \ \\ 
			\textbf{MP} decrypts ES's output and computes the solutions $W = [w_1, \ldots, w_{n-1}]$, $\Lambda = (\lambda_1, \ldots, \lambda_{n-1})$ of the generalized eigenvalue problem $S_W w = \lambda S_B w$ and transforms $v'$ into the LDA space as $v$. Then, it publishes, $W$, $\Lambda$, $v$. \ \\ 
		}  
	} \\
	
	\hline
	\end{tabular} }
	\caption{Privacy-Preserving LDA protocol for enrolling user $u$.}
	\label{fig:pp-lda}
\end{figure*}

\section{Security Analysis}
\label{sec:security}
To show that our BKG technique is secure, we separately prove that it meets the 
BKG requirements from \cite{bal08} -- namely, that 
cryptographic keys are indistinguishable from random 
given the commitment (key randomness), and that given a cryptographic key and a commitment, no 
useful information about the biometric can be reconstructed (biometric privacy). Then, we show that  the PP-LDA protocol is secure against a honest-but-curious adversary.

\if{true}
There are two main aspects for our biometric key generation that need to be addressed. First, the criteria created by Ballard \etal in \cite{bal08} can be used to analyze the security of any BKG, so we discuss them in our context in Section~\ref{sec:invertibility}. Second, the security of the PP-LDA protocol for updating the LDA transformation is analyzed in the honest-but-curious scenario is discussed in Sections~\ref{sec:pp-ldasecurity} and \ref{sec:leakage}.
\fi

\subsection{Key Randomness and Biometric Privacy}
\label{sec:invertibility}

In order to define security of biometric key generation systems, Ballard \etal \cite{bal08} introduced the notions of \emph{Key Randomness} (REQ-KR), \emph{Weak Biometric Privacy} (REQ-WBP) and \emph{Strong Biometric Privacy} (REQ-SBP). 
We assume that the biometric is unpredictable after revealing $l_i$ least significant bits of each feature. Because the least significant bits of $x$ are the most affected by noise, we argue that these bits do not leak information about the $d-l_i$ most significant bits of each feature.

\paragraph{Key Randomness} We formalize the notion of key randomness by defining Experiment $\operatorname{IND-KR}_\A(\kappa)$:

\medskip
\noindent 
{\bf Experiment $\operatorname{IND-KR}_\A(\kappa)$} 

\smallskip
\begin{compactenum}
	\item $\A$ is provided with a challenge $(\PRF_{c_i}(0), \delta)$, $k_b$ and $z$, where $k_0 = \PRF_{c_i}(z)$ and $k_1 \leftarrow_R \{0,1\}^\kappa$ for a bit $b \leftarrow_R \{0,1\}$, corresponding to user $i$.
	\item $\A$ is allowed to obtain biometric information $x_j$ for arbitrary users $j$ such that $j\neq i$.
	\item $\A$ outputs a bit $b'$ as its guess for $b$. The experiment outputs $1$ if $b = b'$, and $0$ otherwise. 
\end{compactenum}

\begin{definition}
We say that a biometric key generation scheme has the Key Randomness property if there exist a negligible function $\operatorname{negl}(\cdot)$ such that for any PPT $\A$, $\operatorname{Pr}[\operatorname{IND-KR}_\A(\kappa) = 1] \leq 1/2+\operatorname{negl}(\kappa)$.
\end{definition}

\begin{theorem} \label{thm:key_randomness}
Assuming that the $\PRF$ is a pseudo-random function family and that biometric $X = (x_1, ..., x_n)$ is unpredictable given $l_i$ least significant bits of each feature $i$, our Fuzzy Commitment scheme has the Key Randomness property.
\label{thm:REQ-KR}
\end{theorem}

\begin{proof}[Proof of Theorem \ref{thm:key_randomness} (Sketch)]
Let $\C$ be a set of codewords such that $|\C| = 2^{\sum_{i=1}^n(d-l_i)}$ and the least 
significant $l_i$ bits of each symbol $i$ of all codewords in $\C$ are 0.
If $c$ is selected uniformly from $\C$, the most significant $d-l_i$ bits in each codeword symbol $i$ are uniformly distributed in $\{0,1\}^{d-l_i}$.
Since $x$ is unpredictable given the least significant $l_i$ bits in each 
feature and the most significant $d-l_i$ bits of each symbol $c_i$ are uniformly distributed, we 
have that $x$ is unpredictable given $\delta$. Because $c = x - \delta$, $c$ is  
unpredictable given $\delta$. 

We now show that any PPT adversary $\A$ that has advantage $1/2+\Delta(\kappa)$ to win the $\operatorname{IND-KR_\A(\kappa)}$ experiment can be used to construct a distinguisher $\D$ that has similar advantage in distinguishing $\PRF$ from a family of uniformly distributed random functions.

$\D$ is given access to oracle $O(\cdot)$ that selects a random codeword $c$ and a random bit $b$, and responds to a query $q$ with random (consistent) values if $b=1$, and with $\PRF_c(q)$ if $b=0$. 
$\D$ selects a random $z$, $c' \leftarrow \C$ and $x' \leftarrow \X$, and sets $\delta' = x' - c'$. Then $\D$ sends $(O(0), \delta'), O(z)$ to $\A$. %

$c$ is unpredictable given the least significant $l_i$ bits of each codeword symbol, and $\delta'$ is chosen from the same distribution as $\delta$. %
If $b=0$, $(\PRF_{c}(0), \delta), \PRF_c(z)$ is indistinguishable from $(\PRF_{c'}(0), \delta'), \PRF_{c'}(z)$, because the $\delta$ and $\delta'$ follow the same distribution, $c$ and $c'$ are unpredictable given $\delta$ and thus the output of both $\PRF_{c}(\cdot)$ and $\PRF_{c'}(\cdot)$ are indistinguishable from random. 
If $b=1$, then $O(\cdot)$ is a random oracle, so $(O(0), \delta'), O(1)$ is indistinguishable from pair $(\PRF_c(0), \delta), \PRF_c(1)$. $\delta$ and $\delta'$ are chosen from the same distribution and $c$ is unpredictable given $\delta$, so $\PRF_c(\cdot)$ is indistinguishable from random. 
 
Eventually, $\A$ outputs $b'$, and $D$ outputs  $b'$ as its guess. It is easy to see that $\D$ wins iff $\A$ wins, so $\D$ is correct with probability $1/2+\Delta(\kappa)$. Therefore, if $\Delta(\cdot)$ is non-negligible, $D$ can distinguish $\PRF$ from a random function with non-negligible advantage over 1/2. However, this violates the security of the PRF; hence, $\A$ cannot exist.
\end{proof}

\paragraph{{\em Weak} and {\em Strong} Biometric Privacy.}
REQ-WBP states that the adversary learns no useful information about a biometric 
signal from the commitment and the auxiliary information, while REQ-SBP states 
that the adversary learns no useful information about the biometric given 
auxiliary information, the commitment and the key.
It is easy to see that, in our BKG algorithms, strong biometric privacy implies 
weak biometric privacy: key $k$ is computed as $\PRF_c(z)$; since the adversary 
knows $\PRF_c(0)$ as part of the commitment, $\PRF_c(z)$ does not add useful 
information. 

We assume that the adversary has access to all public information -- i.e., the 
LDA matrix, the vector of aggregate variances in the LDA space and all system 
parameters --  and user-specific information such as the commitment $(\PRF_c(0), 
\delta)$, a list of keys computed as $k_i = \PRF_{c}(z_i)$ and a list of values 
$z_i$.

Since the output of $\PRF_c$ does not reveal $c$, $\PRF_c(0)$ and $k_i$-s do not 
disclose information about $x$. On the other hand, $\delta$ reveals the least 
significant $l_i$ bits of $x$. However, $x$ is unpredictable given its $l_i$ lest 
significant bits. Therefore, $x$ cannot be reconstructed from $\PRF_c(0)$, 
$k_i$ and $\delta$. Since $c$ is uniformly distributed in $\C$, $\delta$ does not 
reveal information about the most significant $d-l_i$ bits of $x$. Moreover, the $l_i$ least significant bits of $x$ are highly perturbed by noise and therefore do not reveal useful information about the biometric signal.

%
%
%
%

\subsection{Security of LDA Protocol}
\label{sec:pp-ldasecurity}

We argue that the protocol in Figure~\ref{fig:pp-lda} is secure, i.e., that ES 
does not learn any information about a specific user biometric, and that MP only 
learns $S_B$ and $S_W$. 
In particular, the user does not possesses the decryption key for the homomorphic 
encryption, and all messages from ES are encrypted. Since the encryption scheme is  
semantically secure, the user cannot extract any information from the protocol 
execution. 

When interacting with the user, ES's view of the protocol consists in the 
encrypted values from the user, encrypted values from previous runs of the 
protocol and the number of users $v$. The output of the server is $\enc{S_W}$, $
\enc{S_B}$ and $\enc{v'}$. Because of the semantic security of the encryption 
scheme, ES cannot tell if the latter three values are replaced with encryptions of 
random elements. Therefore, the protocol does not reveal any information to ES.

During its interaction with MP, ES does not learn any additional information, 
because it does not receive any message from MP. MP receives encrypted values
$\enc{S_W}$, $\enc{S_B}, \enc{v'}$,
that is able to decrypt. As we argue next, ${S_W}$, $
{S_B}$ and ${v}$ do not leak information about a specific user if computer over a 
{\em set} of users.

\subsection{Information Leakage through $S_W$ and $S_B$}
\label{sec:leakage}

Individual $S_{Wu}$ reveal significant information about a single user's 
biometric. For example, they leak feature variance and correlation between features for $u$. However, by 
averaging all users' {\em scatter within} matrices into $S_{W}$, information 
about single users is no longer available. In particular, the larger the number 
of users, and more uniform their selection, the closer $S_{W}$ will be to the 
value corresponding to the general population, which is assumed to be known.
The same argument applies to both $S_{B}$ and $\enc{\sigma'}$.

However, two tuples of elements $({S_W}^{t_1},{S_B}^{t_1},{\sigma'}^{t_1})$ and $({S_W}^{t_2},
{S_B}^{t_2},{\sigma'}^{t_2})$ generated at different points in time $t_1, t_2$ reveal 
aggregate information corresponding to the users who enrolled between $t_1$ and 
$t_2$. If only a single user enrolls between $t_1$ and $t_2$, then it is possible to 
reconstruct $S_{Wu}$ as ${S_W}^{t_2}-{S_W}^{t_1}$. 
Therefore, in order to prevent this attack, ${S_W}$, ${S_B}$ and ${v'}$ should be updated in batches.

\section{Experimental Evaluation}
\label{sec:evaluation}

\begin{table*}[htbp]
\centering \small
\begin{tabular}{lc|llll|llll}
                          & 					& \multicolumn{4}{c}{keyhold+digraph}       & \multicolumn{4}{c}{keyhold only}           \\
                          & minutes			& entropy 				& FAR    & FRR    & availability & entropy & FAR     & FRR     & availability	\\ \hline
\multirow{2}{*}{with LDA} & 4 						& 99.95\% 		& 5.6\%  & 6.8\% 	& 81.5\%      & 97.7\% 	& 6.9\%  & 14.1\% & 94.4\%  		\\
							 & 8 						& 99.95\%		& 5.5\%  & 3.6\% 	& 98.3\%      & 96.8\% 	& 7.7\%  & 8.0\% & 99.6\%  			\\ \hline 
\multirow{2}{*}{w/o LDA}  & 4 						& 81.9\%    	& 9.2\%  & 9.8\% 	& 82.7\%      & 62.1\% & 12.7\% & 15.1\% & 94.4\%      	\\ 
							 & 8 						& 87.6\% 		& 6.6\%  & 6.9\% 	& 99.0\%      & 67.7\% & 11.3\%  & 10.1\% & 99.7\%  		\\
\end{tabular}
\caption{BKG results. Whole training session was used for creating commitments, 4-minute and 8-minute slices from testing session were used for key retrieval. We report availability as the percentage of time slices that contain at least two vectors with all required features.}
\label{fig:results-eer}
\end{table*}

\begin{table*}[ht]
\centering \small
\begin{tabular}{ll|lll|ll}
& & \multicolumn{3}{c}{computation} & \multicolumn{2}{c}{communication} \\           
features & users & user & ES & MP & user-ES & ES-MP \\ \hline
23       & 250   & 13 min 39 s & 40	min	43	s & 4 s & 17 MB & 135 KB \\
23       & 500   & 27 	min	9	s & 81 min	14	s & 4 s & 34 MB & 135 KB \\
31       & 250   & 26 	min	22	s & 78	min	42	s & 8 s & 33 MB & 260 KB \\
31       & 500   & 52	 min	33	s & 157	 min 16	s & 8 s & 65 MB & 260 KB
\end{tabular}
\caption{Computational and communication overhead for PP-LDA protocol.}
\label{fig:pplda-overhead}
\end{table*}

In order to quantify the biometric performance and key reconstruction reliability 
of our BKG technique, we performed free-text typing experiments on 486 volunteer 
subjects. Each subject was asked to answer between 10 and 13 questions, typing 
{\em at least} 300 character in each answer. Data was collected in two separate 
45-to-120 minute sessions, which took place on different days. Experiments were 
performed using a custom Java GUI, which recorded keystroke with a 15.625 ms 
resolution, and on a standard QWERTY keyboard.

We used two feature subsets of features: the 23 most available \emph{keyhold features},\footnote{`Spacebar', `E', `O', `I', `A', `S', `H', `N', `R', `T', `L', `D', `U', `Y', `W', `G', `P', `C', `M', `B', `F', `V', `K'.} and keyhold features supplemented with $9$ most available \emph{digraph features}.\footnote{`HE', `IN', `TH', `ER', `AN', `RE', `EN', `ND', `HA'.} These features were chosen based on the availability in the first session. We then removed outliers by deleting all feature values higher than 500 ms. Finally, we discretized each feature in the range from 0 to 500 ms.

Data from the first session was used to create the commitments (biometric keys). For each user, we obtained per-feature variance and mean from the whole session. We used the mean to commit to the user's cryptographic key (see Section~\ref{sec:techniques}), and the per-user variance to compute the global variance. 

\subsection{False Accept/Reject and Availability}
A standard metric for evaluating biometric systems is \emph{equal error rate} 
(EER), which is defined as the value that FAR and FRR assume when they are equal. 
In our scenario, we have a false reject when a user's biometric fails to decommit 
the user's cryptographic material -- i.e., when the biometric sample is not {\em 
close enough} to the original biometric. False accept is defined as the event when 
a user's biometric can be used to successfully decommit another user's 
cryptographic information.

When dealing with discrete systems, FAR and FRR might never assume the same value. 
In this case, we approximated EER by reporting both FAR and FRR at their minimum 
distance. Entropy is reported as the percentage of maximum Shannon entropy of 
discretized user templates in our dataset available through our BKG algorithm.

To evaluate FAR without LDA, we implemented a \emph{zero-effort} impostor attack. 
This attack consists in employing a user's biometric to decommit other users' 
cryptographic keys. With LDA we used {\em cross-validation} instead of zero-effort 
attacks. We enrolled all users except for one, which we refer to as {\em 
impostor}. We then used impostor biometric data to attempt to decommit enrolled 
users' biometric keys. We repeated this experiments for each user, so that all of 
them could act as impostor. Using impostors that were not enrolled in the system 
gives better chance to succeeding in the attack, as the LDA transformation 
can maximize separation among users that are enrolled in the system.

Results are presented in Table \ref{fig:results-eer} for both keyhold features 
alone and for keyhold with digraph. The results clearly show that using LDA 
improves both FAR/FRR and entropy. With 4-minute slices of testing data and using 
both keyhold and digraph features, LDA improved FAR-FRR from 9.2-9.8\% to 5.6-6.8\%.
With 8-minute slices, the FAR-FRR improved from 6.6-6.9\% to 5.5-3.6\% using LDA.

One important issue to address is availability of the biometric. Each feature used 
for generating the key must also be used for decommitting, as both LDA and our 
error-correcting code cannot handle erasures. Our results show that 4-minute (8-minute)
sample of keystrokes carries all the required information with probability greater than 
94\% (99\%, respectively) for keystroke only, and over 81\% (98\%, respectively) samples have all the 
required keystrokes and digraphs. The FAR/FRR results are improved when more 
keystrokes/digraphs are available for each feature, at the cost of lower 
availability. Both results for 1 minimum sample and 2 minimum samples are provided 
in Table \ref{fig:results-eer}. 

	\subsection{Computational Overhead of PP-LDA}
	The overhead in our privacy-preserving protocol is dominated by encryptions, 
decryptions and operations in the encrypted domain. We instantiated our protocol using the DKG \cite{dam08,dam08cor} cryptosystem with 1024bit key, 160bit subgroup size and 65bit plaintext size. We run our Java single-threaded prototype implementation on a desktop computer with Intel Xeon E5606 CPU at 2.13 GHz with 48 GB RAM running on Windows 7. 

The amount of computation and communication depends on the number of features, indicated with $n$, and the number of enrolled users $m$. During the protocol, a new user performs $O(n^2)$ encryptions and $O(m n^2)$ exponentiations. The enrollment server computes $O(m n^2)$ multiplications and exponentiations, and the matrix publisher does $O(n^2)$ decryptions. The overall amount of communication in both rounds is $O(m n^2)$ between the user and enrollment server and $O(n^2)$ between the ES and MP. 

The overhead of PP-LDA is reported in Table \ref{fig:pplda-overhead}. 
Because both the computation and communication depend on multiple parameters, we report representative parameter combinations corresponding to our settings. 

Our experiments confirm that the cost of the PP-LDA protocol is relatively small, since the protocol is only executed once for each new user.

\section{Related Work}
\label{sec:related}

\paragraph{BKG based on Behavioral Biometrics.}
Monrose \etal~\cite{mon01} evaluate the performance of BKG based on spoken password using data from 50 users. They report a false-negative rate of 4\%. 

Handwritten signature is another behavioral modality, where biometric key generation has been studied. Multiple papers, for example by Freire \etal \cite{fre06}, Feng \etal \cite{fen02} and more recently Scheuermann \etal \cite{sch11} evaluate the performance. The dataset sizes for the first three papers are 330, 25 and 144 users; the last paper does not include the number of users. The false accept/false reject rates presented are 57\% FRR and 1.7\% FAR in \cite{fre06} with skilled forgeries, 8\% EER in \cite{fen02}. 

\paragraph{BKG based on Physical Biometrics}
Physical biometrics have also been used for biometric key generation, evaluated on fingerprints by Clancy \cite{clancy2003secure} \etal, Uludag \etal \cite{fuzzy2}, Sy and Krishnan \cite{sy12} and others. BKG on iris was studied by Rathgeb and Uhl \cite{rathgeb2010privacy}, \cite{rathgeb2009iris} and Wu \etal \cite{wu2008novel}, and on face images by Chen \etal \cite{che07}. 

\paragraph{Passwords Hardening}
In \cite{keystroke-BKG}, Monrose \etal use keystroke timing for increasing entropy (or ``hardening'') users passwords. Their technique extracts entropy from keystroke data, but does not use free-text and also does not generate keys.

\paragraph{Privacy-Preserving LDA.}
The problem of computing LDA on horizontally and vertically partitioned data has 
been addressed in~\cite{han08} by Han and Ng. However, their protocol is not 
suitable in our setting. In particular, their technique addresses the problem 
where two parties have different partitions of a dataset, and want to compute a 
joint LDA matrix. In our scenario, however, we have many parties -- i.e., the  
users -- who want to compute a common LDA matrix.

\section{Conclusion}
\label{sec:conclusion}
Biometric key generation is an important and general primitive that can be used -- 
among other things -- for authentication, encryption and access control. In this 
paper we present the first BKG algorithm suitable for continuous authentication, 
based on keystroke dynamics. Our algorithm uses LDA to improve reliability and 
biometric performance. We therefore designed and implemented a secure privacy-
preserving protocol for computing and updating LDA parameters using all users' 
biometric signals. 

Biometric performance and computational overhead of our techniques are evaluated 
on a prototype implementation. Our experiments show that our BKG technique has low 
error rate (between 3.6\% and 5.5\%), and presents limited overhead.

{\footnotesize
\bibliographystyle{ieee}
\bibliography{references}

\begin{thebibliography}{10}\itemsep=-1pt

\bibitem{bal08}
L.~Ballard, S.~Kamara, and M.~Reiter.
\newblock The practical subtleties of biometric key generation.
\newblock In {\em USENIX Security Symposium}, 2008.

\bibitem{che07}
B.~Chen and V.~Chandran.
\newblock Biometric based cryptographic key generation from faces.
\newblock In {\em Digital Image Computing Techniques and Applications}, 2007.

\bibitem{clancy2003secure}
T.~Clancy, N.~Kiyavash, and D.~Lin.
\newblock Secure smartcardbased fingerprint authentication.
\newblock In {\em ACM SIGMM Workshop on Biometrics Methods and Applications},
  2003.

\bibitem{dam08cor}
I.~Damg{\aa}rd, M.~Geisler, and M.~Kr{\o}ig{\aa}rd.
\newblock A correction to efficient and secure comparison for on-line auctions.
\newblock Cryptology ePrint Archive, Report 2008/321, 2008.

\bibitem{dam08}
I.~Damg{\aa}rd, M.~Geisler, and M.~Kr{\o}ig{\aa}rd.
\newblock Homomorphic encryption and secure comparison.
\newblock {\em IJACT}, 2008.

\bibitem{fen02}
H.~Feng and C.~Wah.
\newblock Private key generation from on-line handwritten signatures.
\newblock {\em Information Management \& Computer Security}, 10(4), 2002.

\bibitem{fis36}
R.~Fisher.
\newblock The use of multiple measurements in taxonomic problems.
\newblock {\em Annals of Eugenics}, 7(7), 1936.

\bibitem{fre06}
M.~Freire-Santos, J.~Fierrez-Aguilar, and J.~Ortega-Garcia.
\newblock Cryptographic key generation using handwritten signature.
\newblock In {\em Defense and Security Symposium}, 2006.

\bibitem{han08}
S.~Han and W.~Ng.
\newblock Privacy-preserving linear fisher discriminant analysis.
\newblock In {\em PAKDD}, 2008.

\bibitem{fuzzy_commitment}
A.~Juels and M.~Wattenberg.
\newblock A fuzzy commitment scheme.
\newblock In {\em CCS}, 1999.

\bibitem{lee58}
C.~Lee.
\newblock Some properties of nonbinary error-correcting codes.
\newblock {\em IRE Transactions on Information Theory}, 4(2), 1958.

\bibitem{mon01}
F.~Monrose, M.~Reiter, Q.~Li, and S.~Wetzel.
\newblock Cryptographic key generation from voice.
\newblock In {\em S\&P}, 2001.

\bibitem{keystroke-BKG}
F.~Monrose, M.~Reiter, and S.~Wetzel.
\newblock Password hardening based on keystroke dynamics.
\newblock {\em Int. J. Inf. Sec.}, 1(2), 2002.

\bibitem{unpredictable}
M.~Naor and O.~Reingold.
\newblock From unpredictability to indistinguishability: A simple construction
  of pseudo-random functions from {MAC}s.
\newblock In {\em CRYPTO}, 1998.

\bibitem{rathgeb2009iris}
C.~Rathgeb and A.~Uhl.
\newblock An iris-based interval-mapping scheme for biometric key generation.
\newblock In {\em ISPA}, 2009.

\bibitem{rathgeb2010privacy}
C.~Rathgeb and A.~Uhl.
\newblock Privacy preserving key generation for iris biometrics.
\newblock In {\em Communications and Multimedia Security}, 2010.

\bibitem{rot06}
R.~Roth.
\newblock {\em Introduction to Coding Theory}.
\newblock Cambridge University Press, New York, NY, USA, 2006.

\bibitem{sch11}
D.~Scheuermann, B.~Wolfgruber, and O.~Henniger.
\newblock On biometric key generation from handwritten signatures.
\newblock In {\em BIOSIG}, 2011.

\bibitem{sy12}
B.~Sy and A.~Krishnan.
\newblock Generation of cryptographic keys from personal biometrics: An
  illustration based on fingerprints.
\newblock In J.~Yang and S.~Xie, editors, {\em New Trends and Developments in
  Biometrics}. InTech, 2012.

\bibitem{fuzzy2}
U.~Uludag, S.~Pankanti, and A.~Jain.
\newblock Fuzzy vault for fingerprints.
\newblock In {\em AVBPA}, 2005.

\bibitem{wu2008novel}
X.~Wu, N.~Qi, K.~Wang, and D.~Zhang.
\newblock A novel cryptosystem based on iris key generation.
\newblock In {\em International Conference on Natural Computation}, 2008.

\bibitem{wuh08}
Y.~Wu and C.~Hadjicostis.
\newblock Decoding algorithm and architecture for {BCH} codes under the {L}ee
  metric.
\newblock {\em Transactions on Communications}, 56(12), 2008.

\end{thebibliography}
}

\end{document}